\documentclass[a4paper,UKenglish,cleveref, autoref]{lipics-v2019}

\usepackage{algorithm}
\usepackage[noend]{algorithmic}
\usepackage{multirow} 
\usepackage{amsmath}

\newcommand{\old}[1]{{}} 

\newtheorem{observation}[theorem]{Observation}
\newtheorem{problem}[definition]{Problem}

\newcommand{\N}{{{\cal{N}}}}
\newcommand{\A}{{{\cal{A}}}}

\renewcommand{\P}{{{\cal{P}}}}

\newcommand{\C}{{{\cal{C}}}}

\newcommand{\kdom}{{{$k$-$Dom_V (B)$}}}

\title{A Linear-Time Algorithm for Minimum $k$-Hop Dominating Set of a Cactus Graph} 

\titlerunning{A Linear-Time Algorithm for Minimum $k$-Hop Dominating Set of a Cactus Graph}

\author{A. Karim Abu-Affash}{Software Engineering Department, Shamoon College of Engineering, Beer-Sheva 84100, Israel}{abuaa1@sce.ac.il}{}{}

\author{Paz Carmi}{Department of Computer Science, Ben-Gurion University, Beer-Sheva 84105, Israel}{carmip@cs.bgu.ac.il}{}{}

\author{Adi Krasin}{Department of Computer Science, Ben-Gurion University, Beer-Sheva 84105, Israel}{adikra@post.bgu.ac.il}{}{}

\authorrunning{A.\,K. Abu-Affash, P. Carmi and A. Krasin}

\Copyright{A. Karim Abu-Affash, Paz Carmi and Adi Krasin}



\category{}

\relatedversion{}

\funding{This work was partially supported by Grant 2016116 from the United States -- Israel Binational Science Foundation}

\nolinenumbers 

\EventEditors{John Q. Open and Joan R. Access}
\EventNoEds{2}
\EventLongTitle{42nd Conference on Very Important Topics (CVIT 2016)}
\EventShortTitle{CVIT 2016}
\EventAcronym{CVIT}
\EventYear{2016}
\EventDate{December 24--27, 2016}
\EventLocation{Little Whinging, United Kingdom}
\EventLogo{}
\SeriesVolume{42}
\ArticleNo{23}

\begin{document}

\maketitle

\begin{abstract}
Given a graph $G=(V,E)$ and an integer $k \ge 1$, a $k$-\emph{hop dominating set} $D$ of $G$ is a subset of $V$, such that, for every vertex $v \in V$, there exists a node $u \in D$ whose hop-distance from $v$ is at most $k$. A $k$-hop dominating set of minimum cardinality is called a \emph{minimum $k$-hop dominating set}. In this paper, we present linear-time algorithms that find a minimum $k$-hop dominating set in \emph{unicyclic and cactus graphs.}  To achieve this, we show
that the $k$-dominating set problem on unicycle graph reduces to the \emph{piercing circular arcs problem}, and show 
a linear-time algorithm for piercing sorted circular arcs, 
which improves the best known $O(n\log n)$-time algorithm.  

\end{abstract}


\section{Introduction}
A \emph{dominating set} of an undirected graph $G=(V,E)$ is a subset $D \subseteq V$, such that every vertex in $V\setminus D$ is adjacent to at least one vertex in $D$. For arbitrary graphs, the problem of computing a minimum dominating set is NP-complete~\cite{Garey79}. Since the problem is a special case of minimum set cover, it can be approximated within $1+\log{|V|}$~\cite{Johnson74} but it cannot be approximated within $(1-\varepsilon)\log{|V|}$, for any $\varepsilon>0$~\cite{Raz97}. The problem admits a PTAS on planar graphs~\cite{Baker94} and on unit disk graphs~\cite{Hunt98,Nieberg05}. For trees, cactus, and permutation graphs, the problem can be solved in linear time~\cite{Chao00,Cockayne75,Hedetniemi86}. 

Given a graph $G=(V,E)$ on $n$ vertices and an integer $k \ge 1$, a $k$-\emph{hop dominating set} of $G$ is a subset $D \subseteq V$, such that the hop distance from every vertex in $V\setminus D$ is at most $k$ from at least one vertex of $D$. 
Recently, several works devoted to the problem of finding a minimum $k$-hop dominating set in graphs.
The problem has been shown to be NP-complete~\cite{Amis00,Basuchowdhuri14} in general graphs. 
Demaine et al.~\cite{Demaine05} gave an $O(n^4)$-time (fixed-parameter) algorithm for the problem on planar and map graphs.
Kundu and Majumder~\cite{Kundu16} and Barman et al.~\cite{Barman19} presented an optimal linear-time algorithm for the problem on trees and interval graphs, respectively.
The $k$-\emph{dominating number} of a graph $G$, denoted by $\gamma_k(G)$, is the cardinality of a minimum $k$-\emph{hop dominating set} of $G$. The $k$-dominating number has been widely studied in the literature; see for example~\cite{Davila17,Hansberg07,Henning17,Meierling05,Tian05}.

Applications on dominating and $k$-hop dominating sets in graphs are known in several areas such as wireless networks~\cite{Amis00} and social networks~\cite{Basuchowdhuri14}.

Let $G=(V,E)$ be an undirected graph. $G$ is called a \emph{unicyclic} graph if it contains exactly one cycle. 
A vertex $u \in V$ is called a \emph{cut vertex} if removing $u$ from $G$ increases the number of components of $G$. A \emph{block} of $G$ is a maximal connected subgraph of $G$ without any cut vertex. $G$ is called a \emph{cactus} graph if its blocks are either edges or cycles. 

In \cite{LanC13}, Lan and Chang considered a different variant of $k$-dominating set on unicyclic and cactus graphs, where the goal is to find a minimum dominating set $D$, such that each vertex is dominated by at least $k$ vertices from $D$. 

In this paper, a linear-time algorithm is presented, which computes a minimum piercing set for sorted circular arcs, this improves the current $O(n\log n)$-time algorithm by Katz at al.~\cite{Katz03}. 
We then show how to utilize this result to obtain linear-time algorithms to the minimum $k$-hop dominating set in unicyclic and cactus graphs.


\section{Preliminaries} \label{sec:preli}

Let $G=(V,E)$ be an undirected graph.
For two vertices $u,v \in V$, let $d_G(u,v)$ denote the hop-distance between $u$ and $v$ in $G$.
For a vertex $u \in V$, let $\N_i(u) \subset V$ denote the $i$-neighborhood of $u$, i.e., 
$\N_i(u) = \{v \in V : d_G(u,v) \le i\}$.
For a subset $D \subseteq V$, let $\N_i(D) = \bigcup_{u\in D} \N_i(u)$.
For a subset $D \subseteq V$ and a vertex $u \in V$, let $\delta_D(u)$ denote the distance of $u$ from $D$ (i.e., the distance of $u$ to the closest vertex in $D$). 
A $k$-hop dominating set of $G$ is a subset $D \subseteq V$, such that $\delta_D(u) \le k$, for each vertex $u \in V$.
Throughout the rest of the paper, when we say distance we refer to hop-distance.

\begin{problem}[kHDS]
Given an undirected graph $G=(V,E)$ and an integer $k \ge 1$, the goal in the k-Hop Dominating Set problem is to compute a minimum k-hop dominating set of $G$.
\end{problem}

\subsection{\emph{kHDS} on trees} \label{sec:tree}

Let $T=(V,E)$ be a tree on $n$ vertices. In~\cite{Kundu16}, Kundu and Majumder presented a linear-time algorithm that finds a minimum $k$-hop dominating set of $T$.
For completeness, based on their algorithm, we give a simpler (linear-time) implementation of the algorithm that will be used as a black-box throughout the rest of the paper.

Let $T_r$ be the rooted form of $T$ with an arbitrary vertex $r \in V$. For each vertex $u \in V$, let $T_u$ be the subtree of $T_r$ rooted at $u$ and let $h(T_u)$ denote the height of $T_u$, i.e., 
$h(T_u) = \max_{v \in T_u} d_{T_u}(u,v)$. Let $\pi(u)$ denote the parent of $u$ in $T_r$ and let $child(u)$ denote the set of children of $u$ in $T_u$.   

\begin{observation} \label{obs:heightK}
If $h(T_r) \le k$, then $D=\{r\}$ is a minimum $k$-hop dominating set of $T_r$.
\end{observation}

\begin{observation} \label{obs:easyLife}
There exists a minimum $k$-hop dominating set of $T_r$ that does not contain any leaf of $T_r$.
\end{observation}

Let $B$ be a subset of $V$.
A \emph{partial} $k$-hop dominating set of $B$, denoted by \kdom, is a subset $D$ of $V$, such that 
$B \subseteq \N_k(D)$, i.e., for each vertex $v \in B$, there exists a vertex $u \in D$, such that 
$d_{T_r}(u,v) \le k$.
The partial $k$-hop dominating number, denoted by $\Gamma_V (B)$, is the cardinality of a minimum \kdom. 
%
The following lemma follows from Observation~\ref{obs:easyLife}.
\begin{lemma} \label{lemma:tree1}
  If $T_r$ contains a leaf $v \notin B$, then $\Gamma_V (B) = \Gamma_{V \setminus\{v\}} (B)$.
\end{lemma}
\old{
\begin{proof}
Since $v$ is a leaf in $T_r$, any minimum $k$-$Dom_{V \setminus\{v\}} (B)$ is also a $k$-$Dom_V (B)$. Therefore, $\Gamma_V (B) \le \Gamma_{V \setminus\{v\}} (B)$.
	On the other hand, let $D$ be a minimum $k$-$Dom_V (B)$. If $v \notin D$, then $D$ is also a $k$-$Dom_{V \setminus\{v\}} (B)$. Otherwise, $v \in D$. Then, $D \setminus \{v\} \cup \{\pi(v)\}$ is a $k$-$Dom_{V \setminus\{v\}} (B)$. Therefore, in either cases, $\Gamma_{V \setminus\{v\}} (B) \le |D| = |D \setminus \{v\} \cup \{\pi(v)\}| = \Gamma_V (B)$.
\end{proof}
}
\begin{lemma} \label{lemma:tree2}
Let $T_u$ be a subtree of $T_r$ of height $k$ rooted at $u$. If there exists a leaf $v \in B$ of depth $k$ in $T_u$, 
then $\Gamma_V (B) = \Gamma_{V'} (B') + 1$, 
where $B' = B \setminus \N_k(u)$ and $V'$ is the set obtained from $V$ by removing all vertices of $T_u$.  
\end{lemma}
\begin{proof}
Let $D$ be a minimum $k$-$Dom_{V'} (B')$. Then, by Observation~\ref{obs:heightK}, $D \cup \{u\}$ is a $k$-$Dom_{V} (B)$. Therefore, $\Gamma_V (B) \le |D \cup \{u\}| = \Gamma_{V'} (B') + 1$. 
	On the other hand, since there exists a leaf $v \in B$ of depth $k$ in $T_u$, any minimum $k$-$Dom_{V} (B)$ $D$ must contain at least one vertex $w$ in $T_u$, and $D'=D\setminus \{w\} \cup \{u\}$ is also a minimum $k$-$Dom_{V} (B)$, since each vertex in $B$, that is dominated by $w$, is also dominated by $u$.
	Thus, $D' \setminus \{u\}$ is a $k$-$Dom_{V'} (B')$, since $\N_k(u) \nsubseteq B'$. Therefore, $\Gamma_{V'} (B') \le |D' \setminus \{u\}| = \Gamma_V (B) - 1$.
\end{proof}

The construction and correctness of the algorithm is based on 
Lemma~\ref{lemma:tree1} and Lemma~\ref{lemma:tree2}.
In each iteration, the algorithm finds a subtree $T_u$ of height exactly $k$, adds $u$ to $D$, and removes all the leaves of $T_r$ that are in $\N_k(u)$. 
To implement the algorithm in $O(n)$ time, we use a modified version of the depth first search ($DFS$) algorithm (see Algorithm~\ref{alg:DFS}). 
In each recursive call \emph{ModifiedDFS}$(u)$, we check whether $h(T_u) = k$. If so, we add $u$ to $D$, remove all vertices of $T_u$ from $T_r$ (since they are already dominated by $u$), and return $\delta_D(u)=0$.
Otherwise, we check whether at least one descendant $v$ of $u$ is added to $D$ and $d_T(u,v) + h(T_u) \le k$. If so, we remove all vertices of $T_u$ from $T_r$ and return $\delta_D(u) = \min\{\delta_D(u),d_T(u,v)\}$.
Otherwise, we check again whether $h(T_u) = k$ (in case the descendants reduced the height of $u$ to $k$). If so, we add $u$ to $D$, remove all vertices of $T_u$ from $T_r$, and return $\delta_D(u)=0$. Otherwise ($h(T_u) < k$), we return $\infty$.
\floatname{algorithm}{Algorithm} 
\begin{algorithm}[ht]
\caption{\emph{Solve-kHDS-on-Tree}($T_r$) } \label{alg:tree}

\textbf{Input:} A tree $T_r$ rooted at $r$ \\
\textbf{Output:} A minimum $k$-hop dominating set of $T_r$, $k$-$Dom_{V} (V)$

\begin{algorithmic}[1]

\STATE $D \leftarrow \emptyset$
\STATE \textbf{for} each $u \in V$ \textbf{do} \\
\quad \ compute $h(T_u)$ \\
\quad \ $\delta_D(u) \leftarrow \infty$ 

\STATE $\delta_D(r) \leftarrow$ \emph{ModifiedDFS}$(r)$ \\

\STATE \textbf{if} $\delta_D(r) = \infty$ \textbf{then} \\
\quad \ $D \leftarrow D \cup \{r\}$  \\

\STATE \textbf{return} $D$
	
\end{algorithmic}
\end{algorithm}


\floatname{algorithm}{Procedure}
\begin{algorithm}[htb]
\caption{\emph{ModifiedDFS}($u$) } \label{alg:DFS}
\begin{algorithmic}[1]

\STATE \textbf{if} $h(T_u) = k$ \textbf{then} \\
\quad \ $D \leftarrow D \cup \{u\}$ \\
\quad \ $T_r \leftarrow T_r \setminus T_u$ \\
\quad \ $h(T_u) \leftarrow -1$ \\
\quad \ $\delta_D(u) \leftarrow 0$ \\

\STATE \textbf{else} \\
\quad \ \textbf{for} each $v \in child(u)$ \textbf{do} \\
\quad \ \quad \ \textbf{if} $h(T_v) \ge k$ \textbf{then} \\
\quad \ \quad \ \quad \ $\delta_D(u) \leftarrow \min \{ \delta_D(u), ModifiedDFS(v) + 1 \} $ \\

\quad \ $h(T_u) \leftarrow \max \{ h(T_v) +1 \ : \ v \in child(u) \}$ \ \ \ \ (* updating $h(T_u)$ *) \\
\quad \ \textbf{if} $h(T_u) + \delta_D(u) \le k$ \textbf{then} \\
\quad \ \quad \ $T_r \leftarrow T_r \setminus T_u$  \\
\quad \ \quad \ $h(T_u) \leftarrow -1$ \\

\quad \ \textbf{if} $h(T_u) = k$ \textbf{then} \\
\quad \ \quad \ $D \leftarrow D \cup \{u\}$ \\
\quad \ \quad \ $T_r \leftarrow T_r \setminus T_u$ \\
\quad \ \quad \ $h(T_u) \leftarrow -1$ \\
\quad \ \quad \ $\delta_D(u) \leftarrow 0$ \\
\quad \ \textbf{else} \\
\quad \ \quad \ $\delta_D(u) \leftarrow \infty$ \\

\STATE \textbf{return} $\delta_D(u)$ \\
	
\end{algorithmic}
\end{algorithm}

\old{
\begin{wrapfigure}{R}{0.5\textwidth}
	\vspace{-0.4cm}
  \centering
  \includegraphics[width=0.2\textwidth]{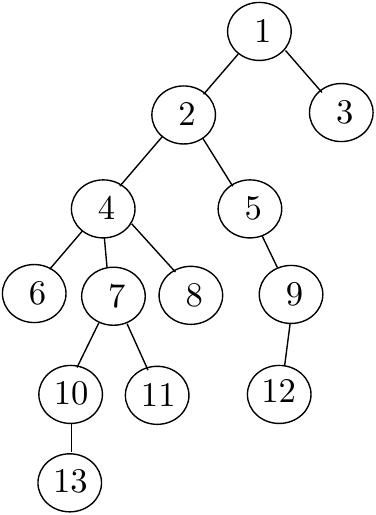}
  \caption{A tree $T_r$ rooted at $r=1$. For $k=2$, Algorithm~\ref{alg:tree} returns $D = \{1,5,7\}$.}
  \label{fig:example}
\end{wrapfigure}
}
\textbf{Example:} Consider the tree $T_r$ rooted at $r=1$ and $k=2$; see Figure~\ref{fig:example}. The first vertex added to $D$ is $7$ in \emph{ModifiedDFS}$(7)$, since $h(T_7)=2$. In this call, we remove $T_7$ from $T_r$, update $h(T_7)=-1$ and return $\delta_D(7)=0$ to \emph{ModifiedDFS}$(4)$. In \emph{ModifiedDFS}$(4)$, we update $\delta_D(4)=1$ and, after traversing vertices 6 and 8, $h(T_4)=1$, and since $h(T_4) + \delta_D(4) =2$, we remove $T_4$ from $T_r$ and return $\delta_D(4)=1$ to \emph{ModifiedDFS}$(2)$. In \emph{ModifiedDFS}$(2)$, since $h(T_2) = 3 > k$, we call \emph{ModifiedDFS}$(5)$ which adds vertex $5$ to $D$, removes $T_5$ from $T_r$, and returns $\delta_D(5)=0$. Then, we update $\delta_D(2) =1$ and $h(T_2)=0$, and since $h(T_2) + \delta_D(2) =1 < k$, we remove $T_2$ from $T_r$ and return $\delta_D(2)=1$ to \emph{ModifiedDFS}$(1)$. In \emph{ModifiedDFS}$(1)$, since $\delta_D(1)=2$ and, after traversing vertex 3, $h(T_1)=1$, we return $\delta_D(1)=\infty$, and therefore, we add vertex $1$ to $D$ as well.
\begin{figure}[htb]
   \centering
       \includegraphics[width=0.2\textwidth]{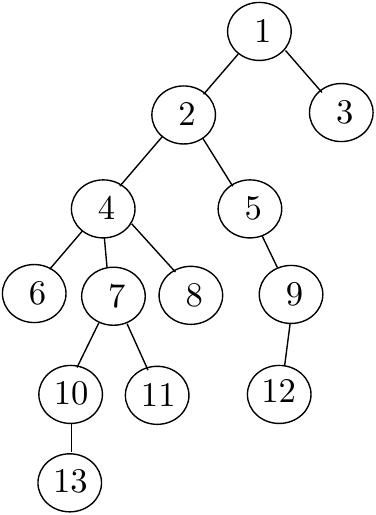}
   \caption{A tree $T_r$ rooted at $r=1$. For $k=2$, Algorithm~\ref{alg:tree} returns $D = \{1,5,7\}$.}              
   \label{fig:example}
\end{figure}

\begin{theorem} \label{thm:tree}
Let $T$ be a tree on $n$ vertices. Then, for any integer $k \ge 1$, one can find a minimum $k$-hop dominating set of $T$ in 
$O(n)$ time.
\end{theorem}

\subsection{Quadratic-time algorithm for unicyclic graphs}
In this section, we present a simple $O(n^2)$-time algorithm that solves \emph{kHDS} on unicyclic graphs (in Section~\ref{sec:unicyclic}, we will improve the running time to $O(n)$).
In~\cite{Kundu16}, Kundu and Majumder proved the following theorem.
\begin{theorem}[\cite{Kundu16}] \label{thm:Kundu}
For each optimal $k$-hop dominating set $D$ of a connected graph $G$, there exists a spanning tree $T$ of $G$, such that $D$ is also an optimal $k$-hop dominating set of $T$.
\end{theorem}

Let $G$ be a unicyclic graph on $n$ vertices and let $C$ be the cycle of $G$. Since there is only one cycle in $G$, there are at most $n$ spanning trees of $G$. For each spanning tree $T$ of $G$, we find a $k$-hop dominating set of $T$ in $O(n)$ time, using Algorithm~\ref{alg:tree}.
Therefore, we can find a $k$-hop dominating set of $G$ in $O(n^2)$ time. The correctness of the algorithm follows immediately from Theorem~\ref{thm:Kundu}.
\begin{theorem}
Let $G$ be a unicyclic graph on $n$ vertices. Then, for any $k \ge 1$, one can find a minimum $k$-hop dominating set of $G$ in $O(n^2)$ time.
\end{theorem}


\subsection{Piercing segments on a line}\label{sec:piercingSegLine}
Let $\A=\{a_1,a_2,\dots,a_n\}$ be a set of $n$ segments on a horizontal line $\ell$. For each $a_i \in A$, let $s_i, e_i$ be the starting and the ending points of $a_i$, respectively, such that $s_i < e_i$. A set $\P$ of points on $\ell$ is a piercing set of $\A$ if for every segment $a_i \in \A$, $a_i \cap \P \neq \emptyset$, that is, for each $a_i \in \A$, there is a point $p \in \P$, such that $s_i \le p \le e_i$.
Assuming $e_i < e_j$, for each $1 \leq i < j \leq n$, a minimum piercing set of $\A$ can be greedily computed in $O(n)$ time.
\old{
 as follows. Set $\A'= \emptyset$. In each iteration we add to $\A'$ the segment $a_i \in \A$, such that $l_i > r_j$, for each point $a_j \in \A'$, and $r_i \le r_k$, for each $a_k \in \A \setminus \A'$. Clearly, for each segment $a_j \in \A$, there exists a segment $a_i \in \A'$, such that $l_j \le r_i \le r_j$, i.e., $r_i$ pierces $a_j$. Therefore, the set $\P = \{r_i \ : \ a_i \in \A'\}$ is a minimum piercing set of $\A$. Since the rightmost endpoints of the segments of $\A$ are sorted, $\A'$ can be computed in $O(n)$ time.
}

\begin{theorem} \label{thm:segments}
Let $\A=\{a_1,a_2,\dots,a_n\}$ be a set of $n$ segments on a horizontal line, such that $e_i < e_j$, for each $1 \leq i < j \leq n$. Then, one can find a minimum piercing set of $\A$ in $O(n)$ time.
\end{theorem}

\section{Piercing Circular Arcs} \label{sec:arcs}

Let $\A=\{a_1=[s_1,e_1],a_2=[s_2,e_2],\dots,a_n=[s_n,e_n]\}$ be a set of $n$ circular arcs located on a circle $C$.
A set $\P$ of points on $C$ is a \emph{piercing set} for $\A$ if, for every arc $a_i \in \A$, $a_i \cap \P \neq \emptyset$.
After sorting the endpoints of the arcs by their polar angle, Katz et al.~\cite{Katz03} showed that a minimum piercing set of $\A$ can be computed in $O(n \log{n})$ time, using a dynamic data structure that supports insertions and deletions to/from $\A$. 
In this section, we present a linear-time algorithm for this problem.

\old{
We assume that no arc $a_i$ is contained in another arc $a_j$, otherwise, we remove $a_j$ from $\A$, since any point that pierces $a_i$ also pierces $a_j$.
Let $\mu$ be a point on $C$. 
For two points $p$ and $q$ on $C$, we say that $p < q$ if and only if $p$ appears before $q$ in clockwise order with respect to $\mu$. 
Let $S=\{s_1,s_2,\dots,s_n \}$ and $F=\{e_1,e_2,\dots,e_n\}$ be the sets of the starting and the ending points of the arcs of $\A$, respectively. 
Let $X=S \cup F = \{p_1,p_2,\dots,p_{2n}\}$ be the set of all the endpoints of the arcs of $\A$, and assume that $p_1 \le p_2 \le \dots \le p_{2n}$. 
}
We assume that no arc $a_i$ is contained in another arc $a_j$, otherwise, we remove $a_j$ from $\A$, since any point that pierces $a_i$ also pierces $a_j$.
Let $S=\{s_1,s_2,\dots,s_n \}$ and $F=\{e_1,e_2,\dots,e_n\}$ be the sets of the starting and the ending points of the arcs of $\A$, respectively.
Let $\mu = s_1$ and let $p_1,p_2,\dots,p_{2n}$ be the points of $X = S \cup F$ ordered in clockwise order with respect to $p_1= \mu$.  For simplicity of presentation, we say that $p_1 \le p_2 \le \dots \le p_{2n}$; see Figure~\ref{fig:piercingArcs}.
\begin{figure}[ht]
   \centering
       \includegraphics[width=0.63\textwidth]{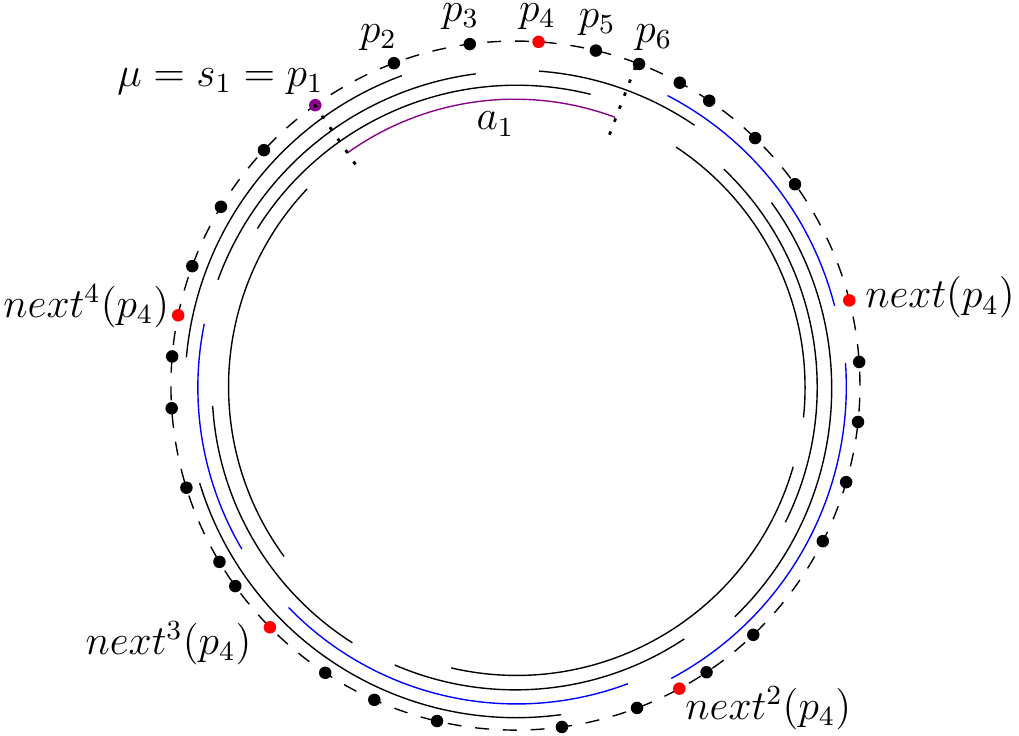}
   \caption{The arcs of $\A$. $X_1 = \{p_1,p_2,p_3,p_4,p_5,p_6\}$. $\P_{\mu}(p_4)$ consists of the red points. }              
   \label{fig:piercingArcs}
\end{figure}

\begin{observation} \label{obs:opt1}
There exists a minimum piercing set of $\A$ that contains only points from the set $F$.
\end{observation}
\begin{proof}
Let $P$ be a minimum piercing set of $\A$. If $P$ contains a point $p \notin F$, then we can replace $p$ by the point $q = \min \{e_i \in F \ : \ e_i > p \}$, and obtain a set $P\cup \{q\} \setminus \{p\}$ which is also a minimum piercing set of $\A$.
\end{proof}

Let $X_1 =X \cap a_1$.
Let $\P$ be a minimum piercing set of $\A$, and notice that $\P$ must contain at least one point from $X_1$. Let $p_j \in \P \cap X_1$. 
Then, by removing all the arcs that are pierced by $p_j$ and considering the remaining  arcs as a set of segments on a line, 
we can find a minimum piercing set of $\A$ in $O(n)$ time, using Theorem~\ref{thm:segments}. 
Therefore, we can find a minimum piercing set of $\A$ in $O(|X_1|\cdot n)$ time.
In the following we show how to implement this algorithm in $O(n)$ time, regardless of the size of $X_1$.

For each point $p_i \in X$, let $\A_i = \{a_j \in \A \ : \ p_i \text{ pierces } a_j\}$ and let $next(p_i)$ be the point $e_j \in F$, such that $s_j > p_i$ and $e_j \le e_k$, for each $a_k \in \A \setminus \A_i$; see Figure~\ref{fig:piercingArcs}.
That is, $e_j$ is the ending point of the first arc in $\A \setminus \A_i$ that appears after $p_i$ (in clockwise order). 

\subsubsection*{Step~1}
For each $p_i \in X$, we compute $next(p_i)$ using Algorithm~\ref{proc:computeNext}. 

\floatname{algorithm}{Procedure}
\begin{algorithm}[ht]
\caption{\emph{Compute-next}($X$, $S$, $F$) } \label{proc:computeNext}
\textbf{Input:} A set $\A$ of arcs sorted in clockwise order by their ending points  \\
\textbf{Output:} $next(p_i)$, for each point $p_i \in X$  

\begin{algorithmic}
\STATE $i \leftarrow 1$
\STATE $j \leftarrow 2$
\STATE \textbf{while} $i < 2n$ \textbf{do} \\
\quad \ \textbf{if} $p_j = s_k \in S$ \textbf{then} \\ 
\quad \ \quad \ $next(p_i) \leftarrow e_k$ \\
\quad \ \quad \ $i++$ \\
\quad \ \textbf{else}  \\
\quad \ \quad \ $j \leftarrow (j+1) \mod 2n$
\end{algorithmic}
\end{algorithm}

\subsubsection*{Step~2} 
For a point $p \in X$ and an integer $i > 1$, let $next^i(p)$ be $next(next^{i-1}(p))$; see Figure~\ref{fig:piercingArcs}.
For two points $p,q \in X$, let $\P_{q}(p) = \{p, next(p), next^2(p), \dots, next^k(p)\}$, such that $next^{k-1}(p) < q \le next^k(p)$.
The following observation follows from the fact that there exists a point $p \in X_1$, such that $\P_{p}(p)$ is a minimum piercing set of $\A$.
\old{
The following observation follows from the fact that every piercing set must contains a point  $p \in X_1$. Moreover, given such a point $p$ the problem 
reduces to piercing segments on a line, where $\P_{p}(p)$ is an optimal solution to it (see~\ref{sec:piercingSegLine}).
}
\begin{observation} \label{obs:opt2}
A minimum piercing set of $\A$ can be computed by taking the minimum over the sets $\P_{p}(p)$, for all $p \in X_1$.
\end{observation}

To this end, we construct a directed graph $G=(X,E)$, such that $E= \{( p, next(p)) \ : \ p \in X\}$. Notice that $|E|= |X|$. 
Recall that $\mu = s_1$ (the starting point of arc $a_1$).
For each point $p \in X_1$, we compute $\P_{\mu}(p)$ recursively using Algorithm~\ref{proc:computeopt}.

\floatname{algorithm}{Procedure}
\begin{algorithm}[ht]
\caption{\emph{ComputeAll}($G=(X,E)$, $\mu$)} \label{proc:computeopt}

\textbf{Input:} A Directed graph $G=(X,E)$ and a reference point $\mu$ \\
\textbf{Output:} $P_{\mu}(p)$, for each $p \in X_1$

\begin{algorithmic}[1]

\STATE \textbf{for} each $p \in X$ \textbf{do} \\
\quad \ $\P_{\mu}(p) \leftarrow \emptyset$  

\STATE \textbf{for} each $p \in X_1$ \textbf{do} \\
\quad \ $\P_{\mu}(p) \leftarrow$ \emph{RecCompute}($p$, $\mu$)  \\

\end{algorithmic}


\underline{\textcolor{white}{zzzzzzZZZZZZZZZZZZZZZZZZZZZZZZZZZZZZZZZZZZZZZZZZZZZZZZZZZZZZZZ}} \\
\underline{\fontsize{9}{12.5}{\sffamily\bfseries Procedure} \emph{RecCopmute}$(p, \mu)$ \textcolor{white}{zzzzzzzzZZZZZZZZZZZZZZZZZZZZZZZZZZZZZZZZZZZ}} 
\begin{algorithmic}[1]
\STATE \textbf{if} $\P_{\mu}(p) \neq \emptyset$ \textbf{then} \\ 
\quad \ \textbf{return} $\P_{\mu}(p)$ \\

\STATE \textbf{if} $next(p) \ge \mu$ \textbf{then} \\
\quad \ $\P_{\mu}(p) \leftarrow \{p\}$ \\
\textbf{else}  \\
\quad \ $\P_{\mu}(p) \leftarrow \{p\} \ \cup$ \emph{RecCompute}($next(p)$, $\mu$) \\

\STATE \textbf{return} $\P_{\mu}(p)$

\end{algorithmic}
\end{algorithm}

For each set $\P_{\mu}(p)$, let $lastNext(p)$ be the last point in $\P_{\mu}(p)$. We check whether $next(lastNext(p)) < p$, which means that the arc ending at $next(lastNext(p))$ is not pierced by $\P_{\mu}(p)$. If so, we add $next(lastNext(p))$ to $\P_{\mu}(p)$.
Finally, we return the set $\P_{\mu}(p)$ of minimum cardinality.
Let $\P$ be the set returned by the algorithm. Thus, by Observation~\ref{obs:opt1} and Observation~\ref{obs:opt2}, $\P$ is a minimum piercing set of $\A$. In the following lemma, we bound the running time of our algorithm.
\begin{lemma} \label{lemma:ComputeTimeCirculeArcs}
Computing $\P$ takes $O(n)$ time.
\end{lemma}

\begin{proof}
It is not hard to see that computing $next(p)$, for each $p \in X$, using Algorithm~\ref{proc:computeNext}, takes $O(n)$ time.
Since each vertex in $G$ has exactly one out-going edge, the number of edges of $G$ is $2n$, and thus, constructing $G$ takes $O(n)$. Moreover, computing $\P_{\mu}(p)$, for all $p \in X_1$, takes $O(n)$, since we traverse each edge in $G$ exactly once during all the calls to Algorithm~\ref{proc:computeopt}. 
Therefore, the running time of the algorithm is $O(n)$.
\end{proof}

The following theorem concludes the result of this section.
\begin{theorem} \label{thm:arcs}
Let $\A=\{a_1,a_2,\dots,a_n\}$ be a set of $n$ circular arcs sorted in clockwise order. Then, one can find a minimum piercing set of $\A$ in $O(n)$ time.
\end{theorem}

\subsection{Piercing circular arcs with additional requirements} \label{sec:arcsAdd}
In this section, we show how to solve two variants of the piercing circular arcs problem, while the first 
variant is interesting by itself, the second is presented here for future use, since it is needed in Section~\ref{sec:cactus}.
In both variants the modifications performed on the algorithm do not affect the running time, and thus it stays linear. 

\paragraph*{Variant-1: Piercing circular arcs with respect to a point $\rho$}
Let $\rho$ be a point in $X$. Given a set $P \subseteq X$, 
%
 let $\delta_P(\rho)$ be the number of points in $X$ between $\rho$ and the closest point of $P$. 
%
%
In this variant, we consider the piercing circular arcs problem with an additional requirement. 
That is, in addition to the set $\A$, we are given a point $\rho \in X$, and the goal is to find a 
minimum piercing set $\P$ for $\A$, such that $\delta_{\P}(\rho) \le \delta_{\P'}(\rho)$, for 
every minimum piercing set $\P'$ for $\A$. 

Assume w.l.o.g., that $\rho \in a_1$.  
We perform the previous algorithm with the following modification. 
When we obtain a minimum piercing set $\P$, we check the distance of the first and the last selected points 
of $\P$ from $\rho$, and return the minimum piercing set $\P$ that minimizes this distance.

By the correctness of the previous algorithm, $\P$ is a minimum piercing set for $\A$, 
and, by the modification, $\delta_{\P}(\rho) \le \delta_{\P'}(\rho)$, for 
every minimum piercing set $\P'$ for $\A$. 

\paragraph*{Variant-2: Piercing circular arcs with respect to a point $\rho$ and a positive integer $k$}
Let $\rho$ be a point in $X$ and let $k$ be an integer. 
For each $1 \le i < k$, let $B_i = \{ a_j=[s_j,e_j] \in \A : \ \delta_{\{s_j\}}(\rho) \ge i \ \wedge \ \delta_{\{e_j\}}(\rho) \ge i \}$; see Figure~\ref{fig:piercingArcs1}. 
Let $\P$ be a minimum piercing set for $\A$, and let $m$ be its size (i.e., $m= |\P|$).
In this variant, in addition to the set $\A$, we are given a point $\rho \in X$ and an integer $k > 1$, and the goal is to find the largest integer $1 \le i < k$ and a piercing set $\P'$ for $\A \setminus B_i$, such that $|\P'| = m -1$, if such a set exists.
\begin{figure}[ht]
   \centering
       \includegraphics[width=0.52\textwidth]{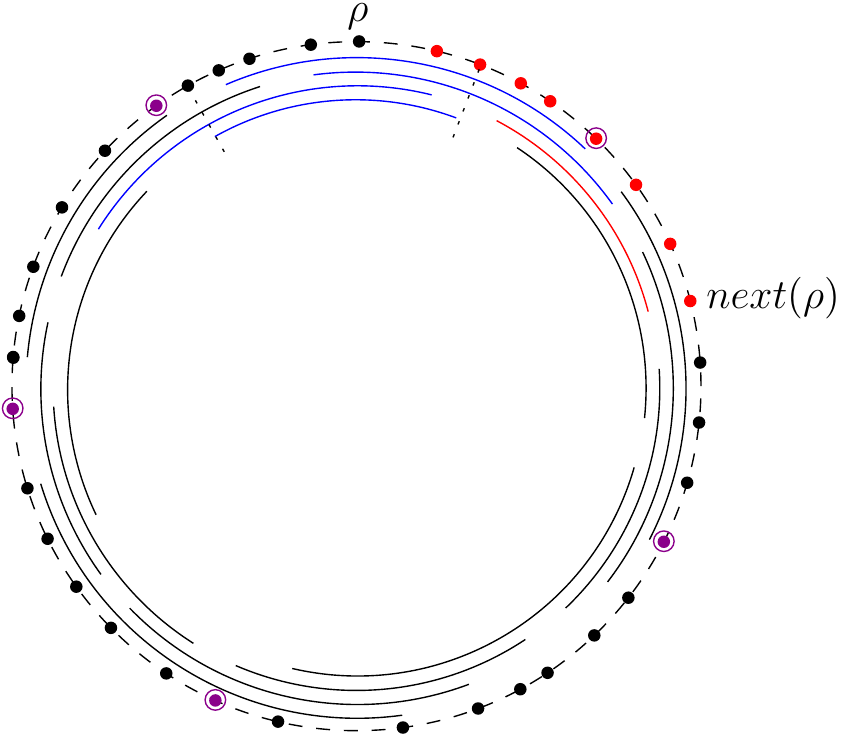}
   \caption{$B_1$ consists of the blue arcs. $X^+$ consists of the red points. A piercing set $\P'$ for $\A \setminus B_i$, for $i=2$, consists of the purple points. }              
   \label{fig:piercingArcs1}
\end{figure}

Let $\A^- = \A \setminus B_1$ and let $X^+ = \{x \in X :  \rho < x \le next(\rho) \}$; see Figure~\ref{fig:piercingArcs1}.
We perform the previous algorithm with the following modifications. 
Instead of going over the points of $X_1$ (in clockwise order), 
we go over the points of $X^{+}$ in sorted order of distance from $\rho$. 
Moreover, instead of requiring to pierce all the arcs in $\A$,  
we require to pierce only $\A^-$.

For each piercing set $\P'_q$ for $\A^-$ of size $m -1$, that starts with a point $q \in X^{+}$, 
we check what is the largest $i$, such that 
$B_1 \setminus A_{\P'_q}  = B_i \setminus A_{\P'_q} $,  where 
$ A_{\P'_q} = \{a \in \A : \ a \ \text{is pierced by a point from } \P'_q \} $.
We return the set $\P'_q$ of size $m -1$ with the largest $i$.

\begin{claim}
The above algorithm finds the largest integer $1 \le i < k$ and 
a piercing set $\P'$ for $\A \setminus B_i$, such that $|\P'| = m -1$, if such a set exists.
\end{claim}
\begin{proof}
Assume such a piercing set exists, and let $\P_{opt}$ be this piercing set.
Let $q_{opt}$ be the starting point in $\P_{opt}$ (i.e., the closest point in $\P_{opt}$ to $\rho$ in clockwise order).
Clearly, $q_{opt} \in X^+$, otherwise, there is an arc in $\A \setminus B_1$ that is not pierced by $\P_{opt}$.
Thus, we are done, since we consider such a solution that starts with $q_{opt}$.
\end{proof}

\section{Solving \emph{kHDS} on Unicyclic Graphs} \label{sec:unicyclic}
In this section, we present a linear-time algorithm that computes a minimum $k$-hop dominating set in unicyclic graphs.
 
Let $G=(V,E)$ be a unicyclic graph on $n$ vertices and let $C$ be the cycle of $G$. Let $r_1,r_2,\dots,r_m$ be the vertices of $C$ ordered in clockwise order with an arbitrary first point $r_1$.
Since $G$ is a unicyclic graph, each $r_i \in C$ is a root of a tree $T_i$. Our algorithm computes a $k$-hop dominating set $D^*$ of $G$ in three steps.

\subsubsection*{Step~1}
For each tree $T_i$ rooted at $r_i$, we compute a $k$-hop dominating set $D_i$ using Algorithm~\ref{alg:tree}. 
For each $1\le i \le m$, let $\delta_{D_i}(r_i)$ and $h(T_i)$ be the distance of $r_i$ from $D_i$ and the height of $T_i$ that are computed at the end of Algorithm~\ref{alg:tree}, respectively.
\begin{itemize} 
	\item If $r_i \notin D_i$, then $1 \le \delta_{D_i}(r_i) \le k$ and $h(T_i) = -1$.
	\item If $r_i \in D_i$, then either $r_i$ was added in Algorithm~\ref{alg:DFS} or in Algorithm~\ref{alg:tree} (Line~4). If $r_i$ was added in Algorithm~\ref{alg:DFS} (when $h(T_i) = k$), then $\delta_{D_i}(r_i) = 0$. Otherwise, $r_i$ was added in Algorithm~\ref{alg:tree} (Line~4, when $0 \le h(T_i) < k$), then we remove $r_i$ from $D_i$ and 
	$\delta_{D_i}(r_i) = \infty$.
\end{itemize}

\subsubsection*{Step~2}
Let $D = \bigcup_{1 \le i \le m} D_i$.
In this step, we consider the roots $r_i \in C$ having $\delta_D(r_i) < k$ in non-decreasing order of $\delta_D(r_i)$. 
For each $r_i$, we consider its neighbors $r_{i-1}$ and $r_{i+1}$. 
If $\delta_D(r_{i-1}) > \delta_D(r_i) +1$ and $h(T_{i-1}) + \delta_D(r_i) +1 \le k$, then we set $\delta_D(r_{i-1}) = \delta_D(r_i) + 1$ and 
update the location of $r_{i-1}$ in the non-decreasing order according to its new $\delta_D(r_{i-1})$ value. 
Similarly, if $\delta_D(r_{i+1}) > \delta_D(r_i) +1$ and $h(T_{i+1}) + \delta_D(r_i) +1 \le k$, then we set $\delta_D(r_{i+1}) = \delta_D(r_i) + 1$ and update the location of $r_{i+1}$ in the non-decreasing order according to its new $\delta_D(r_{i+1})$ value.

\subsubsection*{Step~3}
Let $B=\{u \in C : \delta_D(u) = \infty\}$. 
Recall that a partial $k$-hop dominating set of $B$, denoted by $k$-$Dom_{C} (B)$, is  a subset $D'$ of $C$, such that 
$B \subseteq \N_k(D')$.  
In this step, we compute a minimum $k$-$Dom_{C} (B^+)$ by a reduction to the problem of piercing circular arcs (Section~\ref{sec:arcs}), where $B^+$ is the set containing the vertices of $B$ with their tag along trees.
Notice that, each $r_i \in B$ has $0 \le h(T_i) < k$. 
For each $r_i \in B$, let $m_i=k-h(T_i)$. 
We produce a circular arc $a_i \subseteq C$ containing the vertices $\{r_{i-m_i},r_{i-m_i+1},\dots,r_i,r_{i+1},\dots,r_{i+m_i}\}$ of $C$, where each of them dominates the vertices of $T_i$; see Figure~\ref{fig:arcs}. Let $\A$ be the set of the produced arcs. 
Therefore, to compute a minimum $k$-$Dom_{C} (B^+)$, it is sufficient to compute a minimum piercing set of $\A$.
We compute a minimum piercing set $\P$ of $\A$, using Theorem~\ref{thm:arcs}. At the end of this step, we return the set $D^* = D \cup \P$.
\begin{figure}[ht]
   \centering
       \includegraphics[width=0.44\textwidth]{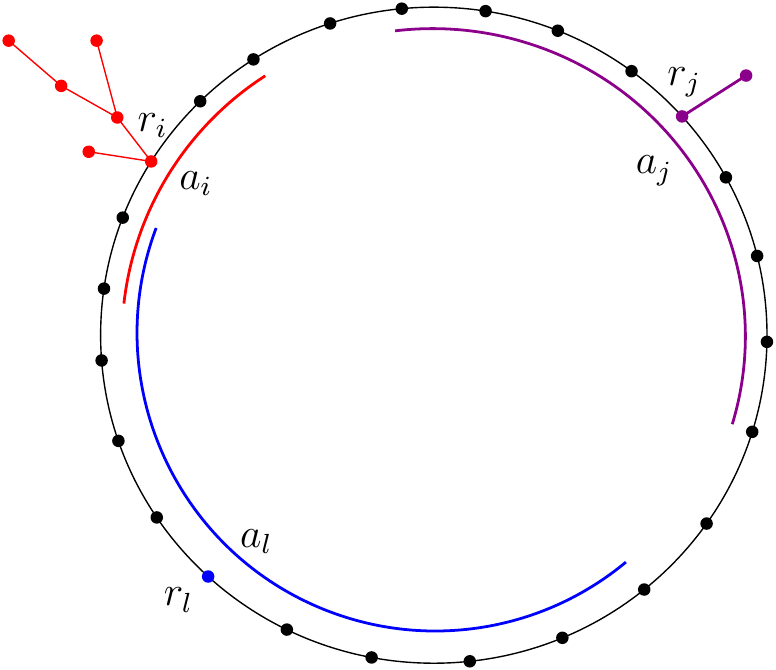}
   \caption{The arcs produced in Step~3, for $k=5$. }              
   \label{fig:arcs}
\end{figure}

\begin{lemma} \label{lemma:CorrectnesUnicycle}
$D^*$ is a minimum $k$-hop dominating set of $G$.
\end{lemma}
\begin{proof}
After applying Step~1, each $D_i$ is a minimum $k$-hop dominating set of $T_i$. Observe that each vertex added to $D_i$ (in Algorithm~\ref{alg:tree}) is as close to the root $r_i$ as possible. Therefore, among all the minimum $k$-hop dominating sets of $T_i$, $D_i$ is the one that minimizes $\delta_{D_i}(r_i)$. 
Moreover, if $r_i \in D_i$ and $r_i$ was added in Algorithm~\ref{alg:tree} (Line~4), then $r_i$ does not have to be in $D^*$, since $0 \le h(T_i) < k$ and $T_i$ can be dominated by vertices of $G \setminus T_i$ of distance $k-h(T_i)$ from $r_i$. In this case, $\delta_{D_i}(r_i) = \infty$ and we remove $r_i$ from $D_i$ at the end of Step~1. 

In Step~2, we update the distances of the vertices in $\N_k(D)$, where $D = \bigcup_{1 \le i \le m} D_i$.
If $\delta_D(r_i)$ has been updated in Step~2, then all the vertices of $T_i$ are dominated by $D$.
Let $T=\bigcup_{r_i \in B} T_i$ be the set of the subtrees that are not dominated yet. 
At this point, $D$ is a minimum $k$-hop dominating set of $G \setminus T$, such that the vertices of $D$ are as close to the vertices of $B$ as possible among all minimum $k$-hop dominating sets of $G \setminus T$. 
Thus, it remains to select a minimum subset of $C$ that dominates the vertices of $T$. 
By the construction of the arcs in Step~3, each vertex in arc $a_i$ dominates the vertices of $T_i$. 
Hence, the minimum piercing set $\P$ that is computed in Step~3 is a minimum $k$-hop dominating set of $T$. 
Therefore, $D^* = D \cup \P$ is a minimum $k$-hop dominating set of $G$.
\end{proof}

\begin{lemma} \label{lemma:ComputeTimeUnicycle}
Computing $D^*$ takes $O(n)$ time.
\end{lemma}
\begin{proof}
Step~1 takes linear time (by Theorem~\ref{thm:tree}). 
We implement Step~2 in linear time as follows. We maintain $k+1$ lists $L_0, L_1, \dots, L_{k-1}$, and $L_{\infty}$, such that $r_i \in L_j$ if and only if $\delta_D(r_i) = j$.
We go over the vertices in the lists $L_0, L_1, \dots, L_{k-1}$ in the following way. 
We consider a vertex $u \in L_i$ before a vertex $v\in L_j$ if $i < j$.
Moreover, when a vertex $u \in L_i$ updates $\delta_D(v)$ of a vertex $v$ to $i+1$, we move $v$ from its current list to $L_{i+1}$.

In Step~3, the arcs in $\A$ are created in sorted order in clockwise order in linear time, since the vertices of the cycle $C$ are sorted in clockwise order. 
Thus, by Theorem~\ref{thm:arcs}, computing a minimum piercing set of $\A$ takes $O(|\A|)$ time, where $|\A| \le |C| \le n$. Therefore, $D^*$ can be computed in $O(n)$ time.
\end{proof}

The following theorem follows from Lemma~\ref{lemma:CorrectnesUnicycle} and Lemma~\ref{lemma:ComputeTimeUnicycle}.
\begin{theorem}\label{thm:unicyclic}
Let $G$ be a unicyclic graph on $n$ vertices. Then, for any $k \ge 1$, one can find a minimum $k$-hop dominating set of $G$ in $O(n)$ time.
\end{theorem}

\subsection{\emph{kHDS} on unicyclic graphs with additional requirements} \label{sec:varUni}
In this section, we consider two variants of the \emph{kHDS} problem on unicyclic graphs. 
In these variants, in addition to a unicyclic graph $G$, we are given a point $\rho \in C$, where $C$ is the cycle of $G$.
In both variants, we apply the \emph{kHDS} algorithm on unicyclic graphs (Section~\ref{sec:unicyclic}) 
with modifications that do not affect its running time.

\paragraph*{Variant-1}
In this variant, the goal is to compute a minimum $k$-hop dominating set $D^*$ of $G$, that minimizes $\delta_{D^*}(\rho)$, among all minimum $k$-hop dominating sets of $G$. 
We apply Steps~1-3 in our algorithm, except that in Step~3, we do not create an arc for $\rho$. Moreover, instead of calling the algorithm of piercing circular arcs (Section~\ref{sec:arcs}), we call the algorithm that solves Variant-1 of piercing circular arcs with additional requirements (Section~\ref{sec:arcsAdd}). 

The correctness of this algorithm follows immediately from Theorem~\ref{thm:unicyclic} and from the correctness of Variant-1 of piercing circular arcs with additional requirements (Section~\ref{sec:arcsAdd}).

\paragraph*{Variant-2}
Let $D^*$ be a minimum $k$-hop dominating set of $G$ and let $m$ be its size.
In this variant, the goal is to compute a $k$-hop dominating set $D$ of $G \setminus \N_k(\rho)$ of size $m-1$, such that the distance from $\rho$ to the farthest vertex in $G \setminus \N_k(D)$ is minimized. 
If there is no $k$-dominating set for $G \setminus \N_k(\rho)$ of size $m-1$, then return $D^*$.
We apply Steps~1-3 in our algorithm, except that, in Step~3, we do not create an arc for $\rho$. Moreover, instead of calling the algorithm of piercing circular arcs (Section~\ref{sec:arcs}), we call the algorithm that solves Variant-2 of piercing circular arcs with additional requirements (Section~\ref{sec:arcsAdd}). 

The correctness of this algorithm follows immediately from Theorem~\ref{thm:unicyclic} and from the correctness of Variant-2 of piercing circular arcs with additional requirements (Section~\ref{sec:arcsAdd}).

\section{Solving \emph{kHDS} on Cactus Graphs} \label{sec:cactus}

In this section, we present a linear-time algorithm that computes a minimum $k$-hop dominating set in cactus graphs. 

Let $G$ be a cactus graph on $n$ vertices and let $C = \{r_1,r_2,\dots,r_m\}$ be a cycle in $G$. 
For each $1 \le i \le m$, let $G_{r_i}$ be the subgraph of $G$ containing $r_i$ and obtained by removing the edges $(r_{i-1}, r_i)$ and $(r_i, r_{i+1})$ from $G$. 
In Algorithm~\ref{alg:cactus}, we first compute a minimum $k$-hop dominating set $D_i$ for each $G_{r_i}$, such $\delta_{D_i}(r_i)$ is minimized, among all minimum $k$-hop dominating sets for $G_{r_i}$.
Then, we compute a minimum $k$-hop dominating set $D$ of the remaining (unicyclic) graph by applying the \emph{kHDS} algorithm that solves the problem in unicyclic graph (Section~\ref{sec:unicyclic}), and we return $D^* = \bigcup_i D_i \cup D$.
\floatname{algorithm}{Algorithm}
\begin{algorithm}[htb]
\caption{\emph{Solve-kHDS-on-Cactus}($G$, $C$)} \label{alg:cactus}

\textbf{Input:} A cactus graph $G$ and a cycle $C$ in it \\
\textbf{Output:} A minimum $k$-hop dominating set $D^*$ for $G$

\begin{algorithmic}[1]
\STATE \textbf{for} each $r_i \in C$ \textbf{do} \\
\quad \ $L_i \leftarrow$  \emph{DFS-Based}($G_{r_i}, r_i$) \\
\quad \ $D_i \leftarrow \emptyset$ 

\STATE \textbf{for} each $r_i \in C$ \textbf{do} \\ 
\quad \  \textbf{if}  $L_i = \emptyset$  \textbf{then}    \hfill {(* $G_{r_i}$ is a tree*)} \\
 \quad \ \quad \   $D_i \leftarrow$ \emph{Solve-kHDS-on-Tree}$(G_{r_i})$ \hfill {(* Algorithm 1, with $r_i$ as the root of $G_{r_i}$ *)} \\
 \quad \  \textbf{else} \\
\quad \ \quad \ \textbf{for} each $\rho \in L_i$ \textbf{do}  \hfill {(* Consider the vertices in sorted order *)} \\
\quad \ \quad \ \quad \ $G_{\rho} \leftarrow $ the subgarph of $G_{r_i}$ obtained by removing the edge $(\rho, \rho.\pi)$ and contains~$\rho$ \\ 
\quad \ \quad \ \quad \ $D_{i} \leftarrow \ D_i \ \cup $ \emph{Solve-Special-kHDS-on-Unicycle}$(G_{\rho}, \rho)$ \\
 \STATE  $D^* \leftarrow \ \bigcup_i D_i$ \\

 \STATE $D \leftarrow$  call solving \emph{kHDS} on unicyclic graphs with the updated graph $G$ (Section~\ref{sec:unicyclic})
 \STATE  $D^* \leftarrow  D^* \cup D$

\STATE \textbf{return} $D^*$
 	
\end{algorithmic}
\end{algorithm}

Given a vertex $r_i \in C$, we compute a minimum $k$-hop dominating set $D_i$ for $G_{r_i}$ as follows.
If $G_{r_i}$ is a tree, then we compute $D_i$ using Algorithm~\ref{alg:tree} (Section~\ref{sec:tree}).
Otherwise, $G_{r_i}$ contains a cycle. 
We perform a DFS-based scan (Algorithm~\ref{alg:dfsbased}) on $G_{r_i}$ that starts from $r_i$ and finds all the cycles of $G_{r_i}$. 
That is, every back-edge to a vertex $u$, found in the DFS-based algorithm, corresponds to a cycle containing $u$ as its representative. We create an ordered list $L_i$ of all the vertices having back-edges in $G_{r_i}$, such that vertex $u$ appears before vertex $v$ in $L_i$, if and only if the finishing time of $u$ in the DFS-base algorithm is less than the finishing time of vertex $v$.     

For a vertex $u \in L_i$, let $G_u$ be the subgraph of $G_{r_i}$ containing $u$ and obtained by removing the edge $(u, u.\pi)$ from $G_{r_i}$.
For each vertex $\rho \in L_i$ (in sorted order), we compute a (nearly) minimum $k$-hop dominating set for $G_{\rho}$ (using Algorithm~\ref{alg:specUni}) as follows. Let $\C$ be the set of cycles of $G_{\rho}$ obtained by the back-edges to vertex $\rho$ in the DFS-Based algorithm. For each cycle $C \in \C$, let $C^+$ be the unicyclic graph corresponding to $C$. 
We first compute a minimum $k$-hop dominating set $D_1$ of $C^+$ that minimizes $\delta_{D_1}(\rho)$, among all minimum $k$-hop dominating sets of $C^+$ (using Variant-1 of solving \emph{kHDS} on unicyclic graphs with additional requirements, Section~\ref{sec:varUni}). Then, we check whether there exists a $k$-hop dominating set of $C^+ \setminus \N_k(\rho)$ of size $|D_1| - 1$. 
\begin{itemize}
	\item If there exists such a set, then we compute a $k$-hop dominating $D_2$ set of $C^+ \setminus \N_k(\rho)$ of size $|D_1| - 1$, such that the distance from $\rho$ to the farthest vertex in $C^+ \setminus \N_k(D_2)$ is minimized (using Variant-2 of solving \emph{kHDS} on unicyclic graphs with additional requirements, Section~\ref{sec:varUni}). Let $u$ be the farthest vertex from $\rho$ in $C^+ \setminus \N_k(D_2)$. We replace $C^+$ in $G$ by the path from $\rho$ to $u$, and we add $D_2$ to $D_i$.
	\item Otherwise, let $i = \delta_{D_1}(\rho)$ and let $u$ be a vertex in $D_1$ of distance $i$ from $\rho$. We replace $C^+$ in $G$ by the path from $\rho$ to $u$ and a new path of length $k$ incident to $u$; see Figure~\ref{fig:CaseOfvariant1}. Moreover, we add $D_1 \setminus \{u\}$ to $D_i$.
\end{itemize}
\begin{figure}[hbt]
   \centering
       \includegraphics[width=0.8\textwidth]{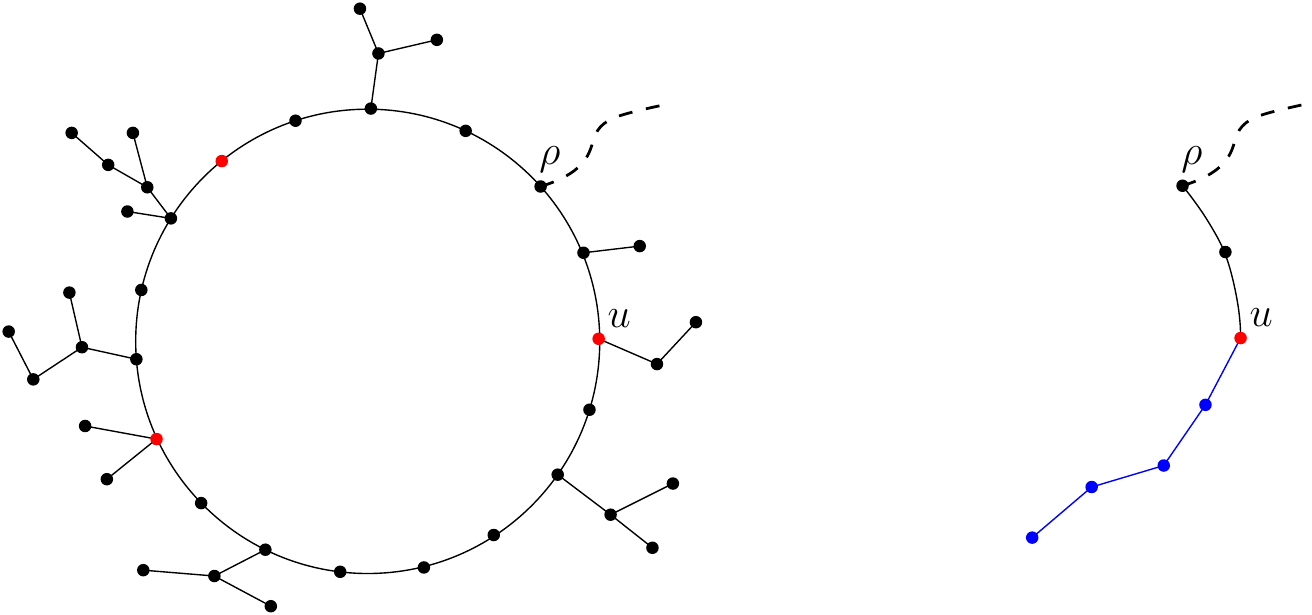}
   \caption{Left, a unicyclic graph $C^+$ in which there is no $4$-dominating set for $C^+ \setminus \N_4(\rho)$ of size $m-1$ (where $m=3$). The vertices of the $4$-dominating set are depicted in red.  
   Right, we replace $C^+$ in $G$ by the path from $\rho$ to $u$ and a new path of length $4$
 attached to $u$,  depicted in blue.  }              
   \label{fig:CaseOfvariant1}
\end{figure}

\floatname{algorithm}{Procedure}
\begin{algorithm}[hbt]
\caption{\emph{Solve-Special-kHDS-on-Unicycle}($G_{\rho}, \rho$)} \label{alg:specUni}

\textbf{Input:} A graph $G_{\rho}$ and a point $\rho \in G_{\rho}$ \\
\textbf{Output:} A (nearly) $k$-hop dominating set for $G_{\rho}$

\begin{algorithmic}[1]
\STATE $\C \leftarrow$ the set of cycles of $G_{\rho}$ obtained by the back-edges in DFS-Based algorithm

\textbf{for} each $C \in \C$  \textbf{do} \\
\quad \ $C^+ \leftarrow$ the unicyclic graph corresponding to $C$ \\
\quad \ $D_1 \leftarrow$ the $k$-hop dominating set computed by Variant-1 of solving \emph{kHDS} on unicyclic \\
\quad \ \quad \ \quad \ \ graphs with additional requirements (Section~\ref{sec:varUni}) with $C^+$ and $\rho$   \\
\quad \ $D_2 \leftarrow$ the $k$-hop dominating set computed by Variant-2 of solving \emph{kHDS} on unicyclic \\
  \quad \ \quad \ \quad \ \  graphs with additional requirements (Section~\ref{sec:varUni}) with $C^+$ and $\rho$ \\

\quad \ \textbf{if}  $|D_2| = |D_1| - 1$ \textbf{then} \\ 
\quad \ \quad \ $u \leftarrow $ the farthest vertex in $C^+ \setminus \N_k(D_2)$ \\
\quad \ \quad \ replace $C^+$ in $G$ by the path from $\rho$ to $u$ \\
\quad \ \quad \ \textbf{return} $D_2$ \\

\quad \ \textbf{else} ($|D_1| = |D_2|$)  \\
\quad \ \quad \  $i \leftarrow \delta_{D_1}(\rho)$ \\
\quad \ \quad \  $u \leftarrow $ a vertex in $D_1$ of distance $i$ from $\rho$ \\
\quad \ \quad \  replace $C^+$ in $G$ by the path from $\rho$ to $u$ and a new path of length $k$ incident to $u$ \\
\quad \ \quad \  $D_1 \leftarrow D_1 \setminus \{u\}$ \\
\quad \ \quad \ \textbf{return} $D_1$ \\

\end{algorithmic}
\end{algorithm}


\floatname{algorithm}{Algorithm}
\begin{algorithm}[hbt]
\caption{\emph{DFS-Based}($G=(V,E)$, $r$)} \label{alg:dfsbased}

\textbf{Input:} A cactus graph $G$ and a source vertex $r$  \\
\textbf{Output:} A list $L$ of vertices that have a back-edge sorted by their finishing time  \\
  \quad \ (the time they are colored black)

\begin{algorithmic}[1]

\STATE $L \leftarrow \emptyset$ \\ 
\STATE \textbf{for} each $u \in V$ \textbf{do} \\
\quad \   $u.color \leftarrow \textsc{white} $ \\
\quad \   $u.\pi  \leftarrow \textsc{null}$\\
\quad \   $u.backEdge \leftarrow \textsc{false} $\\

\STATE  \emph{DFS-visit}($G$, $r$, $L$) \\
\STATE \textbf{return} $L$

\end{algorithmic}


\underline{\textcolor{white}{zzzzzzZZZZZZZZZZZZZZZZZZZZZZZZZZZZZZZZZZZZZZZZZZZZZZZZZZZZZZ}} \\
\underline{\fontsize{9}{12.5}{\sffamily\bfseries Procedure} \emph{DFS-vist}$(G, u, L)$ \textcolor{white}{zzzzzzzzZZZZZZZZZZZZZZZZZZZZZZZZZZZZZZZZZZZ}} 
\begin{algorithmic}[1]

\STATE $u.color  \leftarrow \textsc{gray}$ \\
\STATE \textbf{for} each $v \in G.Adj[u]$ \textbf{do}   \quad \ (* $G.Adj[u]$ is the list of neighbors of $u$ *)\\
\quad \ \textbf{if} $v.color = \textsc{white} $ \textbf{then}  \\
\quad \ \quad \ $v.\pi \leftarrow u$ \\
\quad \ \quad \  \emph{DFS-visit}($G$, $v$, $L$) \\
\quad \ \textbf{else} ($v.color = \textsc{gray}$) \\
  \quad \ \quad \  $v.backEdge \leftarrow \textsc{true}$ 

\STATE \textbf{if} $u.backEdge = \textsc{true} \ $ \textbf{then}  \\
\quad \  add $u$ to the end of $L$ \\

\STATE $u.color \leftarrow \textsc{\textsc{black}}$\\

\end{algorithmic}
\end{algorithm}

\begin{lemma} \label{lemma:CorrectnesCactus}
$D^*$ is a minimum $k$-hop dominating set of $G$.
\end{lemma}
\begin{proof}
The proof is by induction on the number of cycles in $G$.
Let $c$ be the number of cycles in $G$. 
To start the induction, if $c=0$, then, by Theorem~\ref{thm:tree}, the lemma holds, 
and if $c=1$, then, by Theorem~\ref{thm:unicyclic}, the lemma holds. 

Let $OPT$ be a minimum $k$-hop dominating set of $G$.
Let $\rho$ be the first vertex in list $L_1$, that is returned from 
Algorithm \emph{DFS-based} (Algorithm~\ref{alg:dfsbased}). 
Let $C$ be the first cycle considered by \emph{Solve-Special-kHDS-on-Unicycle}($G_{\rho}, \rho$) 
(Algorithm~\ref{alg:specUni}), and let  $C^+$ be the unicyclic graph corresponding to $C$ in $G_{\rho}$. 
Let $OPT_{C^+} = OPT \cap C^+$, and let $m_{opt}=|OPT_{C^+}|$.
Let $D^*_{C^+} = D^* \cap C^+$, and let $m_D =| D^*_{C^+}|$.
From the correctness of  Variant-1 and Variant-2 of Section~\ref{sec:varUni}, we have  $m_D \le m_{opt}$.
Therefore, from the optimality of $OPT$, we have either $m_D = m_{opt}$ or $m_D = m_{opt} -1 $.  
We distinguish between the following cases:
	\begin{itemize}
		\item $C^+ \subseteq  N_k(D^*_{C^+})$ and  $m_D = m_{opt}$. \\
			By the correctness of Variant-1 of Section~\ref{sec:varUni}, we have 
	    $\delta_{D^*_{C^+}}(\rho) \le  \delta_{OPT_{C^+}}(\rho)$.  
			Thus, $\N_k(OPT_{C^+})  \subseteq \N_k(D^*_{C^+})$.
      Let $u$ be the closest vertex to $\rho$ in $D^*_{C^+}$. 
			We remove $u$ from $D^*_{C^+}$, and replace $C^+$ in $G$ by the path from $\rho$ to $u$ and a new path 
			of length $k$ incident to $u$.
			This procedure guaranties that $u$ will be selected again to $D^*$, and $G$ has one less cycle.
					Therefore, by the induction hypothesis, the set $ D^* \setminus D^*_{C^+}$ is a minimum $k$-hop 
					dominating set for the remaining graph $G \setminus \N_k(D^*_{C^+})$.
			\item $C^+ \subseteq  N_k(D^*_{C^+})$ and  $m_D = m_{opt} - 1$. \\ 
    		 Let $u$ be the farthest vertex from $\rho$ in $C^+ \setminus \N_k(D^*_{C^+})$. 
				 Notice that any vertex not in $C^+$ that $k$-dominates $u$ also $k$-dominates all vertices of 
				$C^+ \setminus \N_k(D^*_{C^+})$.
         We  replace $C^+$ in $G$ by the path from $\rho$ to $u$. Thus, $G$ has one less cycle.
				 By adding $\rho$ to $D^*_{C^+}$, we have $\N_k(OPT_{C^+})  \subseteq \N_k(D^*_{C^+})$.
				 Therefore, by the induction hypothesis, the set $ D^* \setminus D^*_{C^+}$ is a minimum $k$-hop dominating 
				 set for the remaining graph $G \setminus \N_k(D^*_{C^+})$.	
		\item $C^+ \nsubseteq  N_k(D^*_{C^+})$ and  $m_D = m_{opt}$. \\
        		Let $u$ be the farthest vertex from $\rho$ in $C^+ \setminus \N_k(D^*_{C^+})$.
						Notice that any vertex not in $C^+$ that $k$-dominates $u$ $k$-dominates all vertices of $C^+ \setminus \N_k(D^*_{C^+})$.
		        We  replace $C^+$ in $G$ by the path from $\rho$ to $u$. Thus, $G$ has one less cycle.
						By the correctness of Variant-2 of Section~\ref{sec:varUni}, $OPT$ also needs to select a vertex from $G$
						of distance at most $k$ from $u$.
						Therefore, by the induction hypothesis, the set $ D^* \setminus D^*_{C^+}$ is a minimum $k$-hop dominating 
					  set for the remaining graph $G \setminus \N_k(D^*_{C^+})$.
		\item $C^+ \nsubseteq  N_k(D^*_{C^+})$ and  $m_D = m_{opt} -1 $. \\
    	      By the optimality of $OPT$, this case could not occur. 
			\end{itemize}		
\end{proof}

\begin{lemma} \label{lemma:ComputeTimeCactus}
Computing $D^*$ takes $O(n)$ time.
\end{lemma}

\begin{proof}
We show that each step in the algorithm takes linear time.
Performing \emph{DFS-Based} on $G_{r_i}$ takes $O(|G_{r_i}|)$ time, for each $r_i \in C$, and $O(n)$ time in total.
Computing a minimum $k$-hop dominating set in trees takes linear time in the size of the tree, and $O(n)$ time in total.
Computing a minimum $k$-hop dominating set in unicyclic graphs takes linear time in the size of the graph.
Computing a minimum $k$-hop dominating set in Variant-1 and Variant-2 of the $k$-hop dominating set in unicyclic graphs takes linear time in the size of the graph, and  $O(n)$ time in total.
\end{proof} 

The following theorem follows from Lemma $\ref{lemma:ComputeTimeCactus}$ and Lemma $\ref{lemma:CorrectnesCactus}$:
\begin{theorem}
Let $G$ be a cactus graph on $n$ vertices. Then, for any $k \ge 1 $, one can compute a minimum $k$-hop dominating set of $G$ in $O(n)$ time.
\end{theorem}



\end{document}